\documentclass[reqno,11pt,a4paper]{amsart}

\usepackage{amssymb}
\usepackage{url}
\usepackage{hyperref}
\usepackage{pdfpages}
\usepackage{graphicx}
\allowdisplaybreaks[4]

\usepackage{mathrsfs}
\let\mathcal\mathscr
\usepackage[mathscr]{eucal}
\usepackage[all]{xy}
\SelectTips{cm}{}

\usepackage{color}

\let\phi=\varphi
\let\kappa=\varkappa

\newcommand*{\sd}[2]{\{\,#1\mid#2\,\}}

\newcommand*{\eval}[1]{\left.#1\right|}
\newcommand*{\abs}[1]{\left|#1\right|}
\newcommand*{\Ev}{\mathbf{E}}

\theoremstyle{theorem}
\newtheorem{proposition}{Proposition}

\theoremstyle{definition}

\theoremstyle{remark}
\newtheorem{remark}{Remark}

\usepackage[all]{xy}
\SelectTips{cm}{}

\usepackage{mathrsfs}
\let\mathcal\mathscr
\usepackage[mathscr]{eucal}

\makeatletter
\newcommand*{\rom}[1]{\expandafter\@slowromancap\romannumeral #1@}
\makeatother
\newcommand{\cprime}{\/{\mathsurround=0pt$'$}}

\begin{document}

\title[The \textsc{ABC}-equation revisited, \rom{1}]{The \textsc{ABC}-equation
  revisited: missed integrable cases. Part \rom{1}: classification and
  (quasi) integrable structures}
\author{I.S.~Krasil'shchik}
\address{Trapeznikov Institute of Control Sciences, 65 Profsoyuznaya street,
  Moscow 117997, Russia}
\email{josephkra@gmail.com}
\thanks{Partially supported by the RSF Grant 25-71-20008}
\begin{abstract}
  We consider the \textsc{ABC}-equation
  $Au_xu_{yt} + Bu_yu_{xt} + Cu_tu_{xy} = 0$, $A$, $B$, $C\in\mathbb{R}$,
  which was proved to be integrable in the case $A + B + C = 0$ (the Veronese
  web equation). Our study reveals ten more cases that admit integrable
  structures that, to the best of our knowledge, where not known before. The
  techniques is based essentially on the theory of the tangent and cotangent
  coverings.
\end{abstract} 
\date{\today}
\maketitle
\setcounter{tocdepth}{3}
\tableofcontents

\section{Introduction}
\label{sec:introduction}

The \textsc{ABC}-equation reads
\begin{equation}
  \label{eq:1}
  Au_xu_{yt} + Bu_yu_{xt} + Cu_tu_{xy} = 0,
\end{equation}
where $A$, $B$, $C$ are nonzero real constants. The equation belongs to the
class of linearly degenerate ones. These equations admit isospectral Lax pairs
which contain differentiation with respect to the spectral parameter,
see~\cite{FM-2015}. In~\cite{KSM-2016}, an infinite series of nonlocal
conservation laws was constructed for Equation~\eqref{eq:1}.

The Lie symmetry analysis distinguishes only one special type, i.e.,
$A + B + C = 0$. In this case,~\eqref{eq:1} is called the Veronese web
equation. The latter was studied in~\cite{Zakh}, where a bi-Hamiltonian pair
was constructed. In~\cite{KMV-2019}, a recursion operator for symmetries was
found and an infinite-dimensional Lie algebra of nonlocal symmetries was
completely described.

Meanwhile, study of conservation laws (that efficiently reduces to the study
of cosymmetries) splits the case $A + B + C \neq 0$ in twelve sub-cases that
are indistinguishable by symmetry analysis. For ten of them we found
B\"{a}cklund transformations between the tangent and cotangent equations. We
call them \emph{nolocal structures} (or \emph{quasi-integrable}) and they
relate (nonlocal) symmetries and cosymmetries of the initial
equation\footnote{Integrability criteria for nonlocal structures are quite
  tricky and will be presented in the forthcoming paper~\cite{KVer-Nonloc}.}.

In Section~\ref{sec:basic-facts}, we present all necessary facts and
constructions and describe our method. For more details we refer the reader
to~\cite{Bocharov-et-al, GeomJetSp, KVV-Springer,
  KV-Trends}. Section~\ref{sec:integrable-cases} contains a description of the
above mentioned nonlocal structures.

The computational details, since they are very voluminous, are hidden in Appendices~\ref{sec:cosymmetries-case-a}--\ref{sec:cons-laws-mathc}.

\section{Basic facts}
\label{sec:basic-facts}

Necessary notation, definitions and results are exposed below.

\subsection{Equations}
\label{sec:equations}

Let $\pi\colon E = \mathbb{R}^m\times\mathbb{R}^n\to\mathbb{R}^n = M$ be a
thrivial bundle and~$\pi_\infty\colon J^\infty(\pi)\to M$ denote the bundle of
its infinite jets. Let also $\mathcal{F}(\pi)$ denote the $\mathbb{R}$-algebra
of smooth functions on~$J^\infty(\pi)$. Denote by $x^1,\dots,x^n$ and
$u^1,\dots,u^m$ coordinates in~$\mathbb{R}^n$ and~$\mathbb{R}^m$,
respectively, and by~$u_\sigma^j$ the corresponding adapted jet coordinates,
where $j = 1,\dots,n$ and~$\sigma$ is a multi-index. The vector fields
\begin{equation*}
  D_i = \frac{\partial}{\partial x^i} + \sum_{j,\sigma}u_{\sigma i}^j
  \frac{\partial}{\partial u_\sigma^j,}
\end{equation*}
on $J^\infty(\pi)$ are called the total derivatives. Linear differential
operators in total derivatives are called $\mathcal{C}$-differential
operators.

Let~$P$ be a free $\mathcal{F}(\pi)$-module of rank~$r$ and
$F = (F^1,\dots,F^r)\in P$. A differential equation in unknowns~$u^j$ is a
submanifold $\mathcal{E}\subset J^\infty(\pi)$ defined by the conditions
\begin{equation*}
  \mathcal{E} = \sd{\theta\in J^\infty(\pi)}{\eval{D_\sigma(F)}_\theta = 0,
   \quad 0\leq \abs{\sigma}\leq\infty},
\end{equation*}
where $D_\sigma$ is the composition of the total derivatives corresponding to
the multi-index~$\sigma$. The algebra of smooth functions on~$\mathcal{E}$ is
defined as $\mathcal{F}(\mathcal{E}) =
\eval{\mathcal{F}(\pi)}_{\mathcal{E}}$. Consequently, any differential
form~$\omega\in\Lambda^k(\mathcal{E})$ is the restriction of a
form~$\bar{\omega}$ on $J^\infty(\pi)$. Restrictions of the total derivatives
to~$\mathcal{E}$ span an involutive distribution~$\mathcal{C}$
on~$\mathcal{E}$ called the Cartan distribution.

Let $\mathcal{O}$ be an $\mathcal{F}$-module and $O\in\mathcal{O}$ be an
element such that the restriction $\eval{O}_{\mathcal{E}}$ vanishes. We shall
always assume that there exists a $\mathcal{C}$-differential operator
$\Delta\colon P\to \mathcal{O}$ such that $O = \Delta(F)$. We also assume that
the equations at hand are $2$-line equations in the sense of
Vinogradov~\cite{C-cpec}.

A symmetry of $\mathcal{E}$ is a $\pi_\infty$-vertical vector field~$S$ that
preserves the Cartan distribution, $[S, \mathcal{C}] \subset
\mathcal{C}$. Denote by~$\kappa$ the $\mathcal{F}(\mathcal{E})$-module of
sections of the pull-back~$\pi_\infty^*(\pi)$. Any symmetry is of the form
\begin{equation*}
  \Ev_\phi = \sum_{j,\sigma}D_\sigma(\phi^j)\frac{\partial}{\partial u_\sigma^j},
\end{equation*}
where $\phi = (\phi^1,\dots,\phi^m)$, the generating section, enjoys the
equation $\ell_{\mathcal{E}}(\phi) = 0$. Here
$\ell_{\mathcal{E}}\colon \kappa\to P$ is the restriction of the linearization
\begin{equation*}
  \ell_F =
  \begin{Vmatrix}
    \sum_\sigma\frac{\partial F^j}{\partial u_\sigma^k}D_\sigma
  \end{Vmatrix}
\end{equation*}
of~$F$ to $P$. We identify symmetries with their generating sections.

The $\mathcal{F}(\mathcal{E})$-module $\Lambda_h^k = \Lambda_h^k(\mathcal{E})$
of horizontal $k$-forms is generated by the elements
$dx^{i_1}\wedge\dots\wedge dx^{i_k}$. The horizontal de~Rham differential
$d_h\colon \Lambda_h^k\to \Lambda_h^{k+1}$ is defined by
\begin{equation*}
  d_h(fdx^{i_1}\wedge\dots\wedge dx^{i_k}) = \sum_iD_i(f)dx^i\wedge
  dx^{i_1}\wedge\dots\wedge dx^{i_k}.
\end{equation*}
An $(n-1)$-form $\omega$ is a conservation law of~$\mathcal{E}$ if
$d_h\omega = 0$. A conservation law is trivial if $\omega = d_h\rho$ for some
$\rho\in \Lambda_h^{n-2}$. Triviality of a conservation law is easily checked
using its generating section (characteristic). To define the latter, let us
introduce the notion of the adjoint operator.

Let $P$, $Q$ be $\mathcal{F}(\mathcal{E})$-modules and $\Delta\colon P\to Q$
be a $\mathcal{C}$-differential operator presented locally in the matrix form
$\Delta =
\begin{Vmatrix}
  \Delta_{ij}
\end{Vmatrix}$, where $\Delta_{ij} = \sum_\sigma a_{ij\sigma}D_\sigma$. Let
$\hat{P} = \hom_{\mathcal{F}(\mathcal{E})}(P, \Lambda_h^n)$ and similarly
for~$Q$. Then the adjoint $\Delta^*\colon\hat{Q}\to \hat{P}$ has the matrix
form $\begin{Vmatrix}
  \Delta_{ij}^*
\end{Vmatrix}^\intercal$, where ${}^\intercal$ denotes the transposition, and
\begin{equation*}
  \Delta_{ij}^* = \sum_\sigma (-1)^{\abs{\sigma}}D_\sigma\circ a_{ij}.
\end{equation*}
Relation between an operator and its adjoint is expressed by the Green formula
\begin{equation}
  \label{eq:2}
  \langle\Delta(p), \hat{q}\rangle - \langle p,\Delta^*(\hat{q})\rangle =
  d_h\omega, 
\end{equation}
where $p\in P$, $\hat{q}\in\hat{Q}$, and $\omega = \omega_\Delta(p,
\hat{q})\in \Lambda_h^{n-1}$.

Now, let $\omega$ be a conservation law and
$\bar{\omega}\in\Lambda^{n-1}(J^\infty(\pi))$ be such that $\omega =
\eval{\bar{\omega}}_{\mathcal{E}}$. Then, since $d_h\omega = 0$, one has
\begin{equation*}
  d_h\bar{\omega} = \Delta(F),\quad \Delta\colon P\to \Lambda_h^n.
\end{equation*}
Set $\psi = \psi_\omega = \Delta^*(1)\in \hat{P}$. This element is the
generating section of~$\omega$ and it equals zero if and only if~$\omega$ is
trivial.

\subsection{Differential coverings and B\"{a}cklund transformations}
\label{sec:diff-cover-backl}

Consider two equations $\mathcal{E}$ and~$\tilde{\mathcal{E}}$ and let
$\tau\colon \tilde{\mathcal{E}}\to \mathcal{E}$ be a smooth map. It is called
a morphism if $d\tau(\mathcal{C}_{\tilde{\theta}}) \subset
\mathcal{C}_{\tau(\tilde{\theta})}$, where
$\tilde{\theta}\in\tilde{\mathcal{E}}$ and $\mathcal{C}_x$ is the Cartan plane
at the point~$x$. A morphism is a (differential) covering
of~$\tilde{\mathcal{E}}$ over~$\mathcal{E}$ if it is a surjection and the
restriction $\eval{d\tau}_{\tilde{\mathcal{C}}_{\tilde{\theta}}}$ is an
isomorphism for any~$\tilde{\theta}$.

Locally, coverings are described as follows. Assume $\tau\colon W =
\mathbb{R}^N \times\mathcal{E}\to \mathcal{E}$, where $N\in \mathbb{N}\cup
\{\infty\}$. Then the tolal derivatives on~$W$ are
\begin{equation*}
  \tilde{D}_i = D_i + X_i,\qquad X_i = \sum_\alpha
  X_i^\alpha\frac{\partial}{\partial w^\alpha},
\end{equation*}
such that
\begin{equation}
  \label{eq:3}
  D_i(X_j) - D_j(X_i) + [X_i, X_j] = 0,\qquad 1\leq i<j\leq N.
\end{equation}
Here $w^j$ are coordinates in $\mathbb{R}^N$ that are called nonlocal
variables. Conditions~\eqref{eq:3} amount to the overdetermined system
\begin{equation*}
  \frac{\partial w^\alpha}{\partial x^i} = X_i^\alpha(\dots x^k\dots
  u_k^j\dots w^\beta\dots) 
\end{equation*}
whose compatibility conditions coinside with~$\mathcal{E}$.

Objects on $\tilde{\mathcal{E}}$ (symmetries, cosymmetries, conservation laws,
etc.) are nonlocal for~$\mathcal{E}$. Any $\mathcal{C}$-differential
operator~$\Delta$ on~$\mathcal{E}$ can be naturally lifted to an
operator~$\tilde{\Delta}$ on~$\tilde{\mathcal{E}}$. Thus, one may consider the
operators~$\tilde{\ell}_{\mathcal{E}}$
and~$\tilde{\ell}_{\mathcal{E}}^*$. Solutions of the equations
\begin{equation*}
  \tilde{\ell}_{\mathcal{E}}(\phi) = 0,\qquad
  \tilde{\ell}_{\mathcal{E}}^*(\psi) = 0
\end{equation*}
are called shadows of symmetries and cosymmetries, respectively. A
B\"{a}cklund transformation between equations~$\mathcal{E}_1$
and~$\mathcal{E}_2$ is the diagram
\begin{equation*}
  \xymatrix{
    &\mathcal{E}\ar[dl]_{\tau_1}\ar[dr]^{\tau_2}&\\
    \mathcal{E}_1&&\mathcal{E}_2\rlap{,}
  }
\end{equation*}
where $\tau_1$, and $\tau_2$ are coverings.

We say that $\omega$ is a two-component conservation law if it is of the form
\begin{equation*}
  \omega = (X_i\,dx^i +
  X_k\,dx^k)\wedge\,dx^{i_1}\wedge\dots\wedge\,dx^{i_{k-2}},\ i\neq k\neq
  i_{\alpha},\ i_1<\dots<i_{k-2},
\end{equation*}
enjoying $D_i(X_k) = D_k(X_i)$. To any such a law there corresponds the
covering~$\tau_\omega$ with the nonlocal variables~$w_\sigma$, where the
multi-index consists of~$i_1,\dots,i_{k-2}$. the defining equations
for~$\tau_\omega$ are
\begin{equation*}
  \frac{\partial w_\sigma^\alpha}{\partial x^s} = D_\sigma(X_s^\alpha),\quad s =
  i, k,\qquad
  \frac{\partial w_\sigma^\alpha}{\partial x^s} = w_{s\,\sigma}^\alpha,\quad
  s\neq i, k.
\end{equation*}

\subsection{The tangent and cotangent coverings}
\label{sec:tang-cotang-cover}

Consider an equation $\mathcal{E} = \{F = 0\}$, $F\in P$. The systems
\begin{equation*}
  \mathcal{T}\mathcal{E}\colon
  \left\{
    \begin{array}{l}
      \ell_{\mathcal{E}}(q) = 0,\\[3pt]
      F = 0;
    \end{array}
  \right.\qquad
  \mathcal{T}^*\mathcal{E}\colon
  \left\{
    \begin{array}{l}
      \ell_{\mathcal{E}}^*(p) = 0,\\[3pt]
      F = 0     
    \end{array}
  \right.
\end{equation*}
are called the tangent and cotangent equations, respectively. The projections
$\mathbf{t}\colon (x, u, q)\mapsto (x, u)$ and
$\mathbf{t^*}\colon (x, u, p)\mapsto (x, u)$ are tangent and cotangent
coverings of the equation~$\mathcal{E}$. It is important to stress that the
variables~$q$ and~$p$ are \emph{odd} of parity~$1$.

The most important properties of these coverings are formulated in the
following two statements.

\begin{proposition}
  \label{sec:tang-cotang-cover-1}
  Let $\mathcal{E}$ be an equation. Its tangent cover possesses the following
  properties
  \begin{enumerate}
  \item \label{item:1} sections of~$\mathbf{t}$ that are morphisms of
    equations are in one-to-one correspondence with symmetries
    of~$\mathcal{E}$\textup{;}
  \item \label{item:2} $\mathbf{t}$-shadows of symmetries linear in~$q_\sigma$
    are identified with $\mathcal{C}$-differential operators
    $\mathcal{R}\colon \kappa\to \kappa$ that take symmetries
    of~$\mathcal{E}$ to symmetries\textup{,} i.e.\textup{,} with recursion
    operators for symmetries\textup{;}
  \item \label{item:3} $\mathbf{t}$-shadows of cosymmetries linear
    in~$q_\sigma$ are identified with $\mathcal{C}$-differen\-tial operators
    $\mathcal{S}\colon \kappa\to \hat{P}$ that take symmetries
    of~$\mathcal{E}$ to cosymmetries\textup{,} i.e.\textup{,} with
    pre-symplectic structures\textup{;}
  \item \label{item:4} $\mathcal{T}\mathcal{E}$ is endowed with a natural
    nilpotent vector field
    $\mathbf{X} = \sum_\sigma q_\sigma\partial/\partial u_\sigma$ of
    parity~$1$.
  \end{enumerate}
  Similarly\textup{,} the cotangent covering enjoys the properties
  \begin{enumerate}
  \item \label{item:5} sections of~$\mathbf{t^*}$ that are morphisms of
    equations are in one-to-one correspondence with cosymmetries
    of~$\mathcal{E}$\textup{;}
  \item \label{item:6} $\mathbf{t^*}$-shadows of symmetries linear
    in~$p_\sigma$ are identified with $\mathcal{C}$-differen\-tial operators
    $\mathcal{H}\colon \hat{P}\to \kappa $ that take cosymmetries
    of~$\mathcal{E}$ to symmetries\textup{,} i.e.\textup{,} with pre-Poisson
    structures\textup{;}
  \item \label{item:7} $\mathbf{t^*}$-shadows of cosymmetries linear
    in~$p_\sigma$ are identified with $\mathcal{C}$-differen\-tial operators
    $\mathcal{R}\colon \hat{P}\to \hat{P}$ that take cosymmetries
    of~$\mathcal{E}$ to cosymmetries\textup{,} i.e.\textup{,} with recursion
    operators for cosymmetries\textup{;}
  \item \label{item:8} for any~$\mathcal{E}$\textup{,} its cotangent equation
    is always an Euler-Lagrange one and consequently carries a natural
    symplectic structure.
  \end{enumerate}
\end{proposition}

\begin{proposition}
  \label{sec:tang-cotang-cover-2}
  There is a one-to-one correspondence between conservation laws
  of~$\mathcal{T}\mathcal{E}$ linear in~$q_\sigma$ and cosymmetries
  of~$\mathcal{E}$. In a similar way\textup{,} There is a one-to-one
  correspondence between conservation laws of~$\mathcal{T}^*\mathcal{E}$
  linear in~$p_\sigma$ and symmetries of~$\mathcal{E}$.
\end{proposition}

\begin{proof}
  The proof is essential for us, because it provides an explicit method to
  construct conservation laws of the tangent and cotangent equations that is
  used in the forthcoming computations.

  Let $\phi$ be a symmetry of~$\mathcal{E}$ and~$\bar{\phi}$ be its extension
  to the space of the ambient jets. Then, by the Green formula~\eqref{eq:2},
  \begin{equation}\label{eq:5}
    \langle\ell_F(\bar{\phi}), p\rangle - \langle\bar{\phi},
    \ell_F^*(p)\rangle = d_h\omega,
  \end{equation}
  and we put $\omega_\phi = \eval{\omega}_{\mathcal{T}^*\mathcal{E}}$. Quite
  similarly, let~$\psi$ be a cosymmetry and~$\bar{\psi}$ be its
  extension. Then
  \begin{equation}\label{eq:4}
    \langle q, \ell_F^*(\psi)\rangle - \langle\ell_F(q), \psi\rangle =
    d_h\omega 
  \end{equation}
  and $\omega_\psi = \eval{\omega}_{\mathcal{T}\mathcal{E}}$.
\end{proof}

\subsection{Nonlocal structures}
\label{sec:nonlocal-structures}

Consider the following three B\"{a}cklund transformations
\begin{equation*}
  \xymatrix{
    &W\ar[dl]\ar_{\tau_1}[dl]\ar[dr]^{\tau_2}&
    &W^\bullet\ar_{\tau_3}[dl]\ar[dr]^{\tau_4}&
    &W^{\bullet\bullet}\ar_{\tau_5}[dl]\ar[dr]^{\tau_6}&\\
    \mathcal{T}\mathcal{E}&
    &\mathcal{T}\mathcal{E}&
    &\mathcal{T}^*\mathcal{E}&
    &\mathcal{T}^*\mathcal{E}\rlap{.}
  }
\end{equation*}
Then, by Proposition~\ref{sec:tang-cotang-cover-1}, the left and the right
diagrams may be considered as nonlocal recursion operators for symmetries and
cosymmetries, respectively, while the middle (a quasi-integrable structure)
one relates (shadows of) symmetries to (shadows of) cosymmetries
of~$\mathcal{E}$ and thus combines some properties of symplectic and Poisson
(Hamiltonian) structures.

Let us explicate our way to construct the above diagrams.

As the first step, we compute symmetries and cosymmetries
of~$\mathcal{E}$. Then, using formulas~\eqref{eq:5} and~\eqref{eq:4} we
construct conservation laws of~$\mathcal{T}\mathcal{E}$
and~$\mathcal{T}^*\mathcal{E}$. For the two-component ones, by the
correspondence described in the end of Subsection~\ref{sec:diff-cover-backl},
coverings over tangent and cotangent equations arise. Let them be
$\tau\colon W\to \mathcal{T}\mathcal{E}$ and
$\tau^*\colon W^*\to \mathcal{T}^*\mathcal{E}$.

Assume that there exist nontrivial solutions to the equations
$\tilde{\ell}_{\mathcal{E}}(\phi) = 0$ and
$\tilde{\ell}_{\mathcal{E}}^*(\psi) = 0$ in~$\tau$. They provide the
B\"{a}cklund transformations
\begin{equation*}
  \xymatrix{
    &W\ar[dl]_\tau\ar[dr]^{\tau_\phi}\\
    \mathcal{T}\mathcal{E}&&\mathcal{T}\mathcal{E}\rlap{,}
  }\qquad
  \xymatrix{
    &W\ar[dl]_\tau\ar[dr]^{\tau_\psi}&\\
    \mathcal{T}\mathcal{E}&&\mathcal{T}^*\mathcal{E}\rlap{.}
  }
\end{equation*}
In a similar way nontrivial solutions to
$\tilde{\ell}_{\mathcal{E}}(\phi) = 0$ and
$\tilde{\ell}_{\mathcal{E}}^*(\psi) = 0$ in~$\tau^*$ provide
\begin{equation*}
  \xymatrix{
    &W^*\ar[dl]_{\tau^*}\ar[dr]^{\tau_\phi}\\
    \mathcal{T}\mathcal{E}&&\mathcal{T}\mathcal{E}\rlap{,}
  }\qquad
  \xymatrix{
    &W^*\ar[dl]_{\tau^*}\ar[dr]^{\tau_\psi}&\\
    \mathcal{T}\mathcal{E}&&\mathcal{T}^*\mathcal{E}\rlap{.}
  }
\end{equation*}

Apply this scheme to Equation~\eqref{eq:1} now.

\section{(Quasi) Integrable cases}
\label{sec:integrable-cases}

Rewrite Equation~\eqref{eq:1} in the form
\begin{equation}
  \label{eq:6}
  u_{xy} = \frac{au_xu_{yt} + bu_yu_{xt}}{u_t},
\end{equation}
where $a = A/C$, $b = -B/C$. The linearization operator is
\begin{align*}
  \ell_{\mathcal{E}} &= D_xD_y- \frac{bu_y}{u_t}D_xD_t- \frac{au_x}{u_t}D_yD_t\\
  &- \frac{au_{yt}}{u_t}D_x - \frac{bu_{xt}}{u_t}D_y + \frac{au_xu_{yt} +
    bu_yu_{xt}}{u_t^2}D_t,
\end{align*}
while the adjoint one is of the form
\begin{align*}
  \ell_{\mathcal{E}}^*
  &= D_xD_y - \frac{bu_y}{u_t}D_xD_t - \frac{au_x}{u_t}D_yD_t\\
  &+ \frac{(a - b )u_tu_{yt} + bu_yu_{tt}}{u_t^2}D_x
    - \frac{(a - b)u_t u_{xt} - au_xu_{tt}}{u_t^2}D_y\\
  &- \frac{(a + b)(au_xu_{yt} + bu_yu_{xt})}{u_t^2}D_t
    - \frac{(a + b) (u_tu_{xt}u_{yt} - au_xu_{tt}u_{yt} -
    bu_yu_{tt}u_{xt})}{u_t^3}. 
\end{align*}
Consequently, the tangent and the cotangent equations are
\begin{align*}
  q_{xy} &= \frac{bu_yq_{xt}}{u_t} + \frac{au_xq_{yt}}{u_t}
  + \frac{au_{yt}q_x}{u_t} + \frac{bu_{xt}q_y}{u_t}
  - \frac{(au_xu_{yt} + bu_yu_{xt}) q_t}{u_t^2}
\intertext{and}
  p_{xy}
  &= \frac{bu_yp_{xt}}{u_t} + \frac{au_xp_{yt}}{u_t}
  - \frac{((a - b)u_tu_{yt} + bu_yu_{tt})p_x}{u_t^2}\\
  &+ \frac{((a - b)u_tu_{xt} - au_xu_{tt})p_y}{u_t^2}
    + \frac{(a + b)(au_xu_{yt} + bu_yu_{xt})p_t}{u_t^2}\\
  &- \frac{(a + b)((au_xu_{yt} + bu_yu_{xt})u_{tt} - u_tu_{xt}u_{yt})p}{u_t^3}.
\end{align*}
respectively.

Computation of symmetries splits Equation~\eqref{eq:6} in two classes:
\begin{enumerate}
\item\label{item:9} $a + b = 1$ ($A + B + C =0$). In this case, the symmetry
  algebra is generated by the elements
  \begin{equation*}
    \phi_0 = f_0,\quad \phi_1 = f_1u_x,\quad \phi_2 = f_2u_y,\quad \phi_3 =
    f_3u_t, 
  \end{equation*}
  where $f_0 = f_0(u)$, $f_1 = f_1(x)$, $f_2 = f_2(y)$, $f_3 = f_3(t)$ are
  smooth functions of their arguments.
\item\label{item:10} $a + b \neq 1$. Here, the generators of the symmetry
  algebra are, $f_1 = f_1(x)$, $f_2 = f_2(y)$, $f_3 = f_3(t)$ together with
  \begin{equation*}
    \phi_{00} = 1,\qquad \phi_{01} = u.
  \end{equation*}
\end{enumerate}

On the other hand, computing cosymmetries of~\eqref{eq:6}, one finds out that
Case~(\ref{item:10}) is divided in~$13$ subclasses that contain~$11$ equation
admitting quasi-integrable structures. They include
\begin{enumerate}
\item $A + B + C = 0$,
\item $A = -C$, $B = -C$,
\item $A = -C$, $B = C$,
\item $A = -C$, $2B = -3C$,
\item $A = -C$, $B = -3C$,
\item $A = -C$, $2B = 3C$,
\item $A = -C$, $B = 3C$,
\item $3A = 2C$, $3B = -2C$,
\item $3A = C$, $3B = -C$,
\item $A = C$, $2B = C$,
\item $A = -2C$, $B = -2C$.
\end{enumerate}
In Subsections~\ref{sec:=-c-b}--\ref{sec:a-=-2c}, we discuss these
equations. Subsection~\ref{sec:a-+-b} is devoted to Case~(\ref{item:9}).

\begin{remark}
  Obviously, the transpositions $x\leftrightarrow y$, $a\leftrightarrow b$
  leave Equation~\eqref{eq:6} invariant, and we identify below the cases that
  are equivalent under this transformation.
\end{remark}

\begin{remark}
  We also omit below the case $a = b = -1$ ($A = B = C$) that admits local
  symplectic and Poisson structures that are a trivial consequence of the fact
  that~\eqref{eq:1} is an Euler-Lagrange equation for these values of the
  parameters.
\end{remark}

\subsection{$A + B + C = 0$}
\label{sec:a-+-b}

Let us discuss this case in detail to illustrate how the procedure works.

\subsubsection{The tangent covering}
\label{sec:tangent-covering}

The cosymmetries of $\mathcal{T}\mathcal{E}$ are presented in
Appendix~\ref{sec:cosymmetries-case-a}, the corresponding conservation law are
given in Appendix~\ref{sec:cons-laws-tang}; to the two-component ones there
correspond the nonlocal variables defined by
\begin{equation}\label{eq:8}
  \begin{array}{rcl}
    w_{1,t}
    &= & \dfrac{f_1u_t}{u_x^2}q_x - \dfrac{af_1}{u_x}q_t
         + \dfrac{(a - 1)f_1u_{xt}}{u_x^2}q,\\
    w_{1,y}
    &= & -\dfrac{(a - 1)f_1\big((a - 1)u_yu_{xt} -
         au_xu_{yt}\big)}{u_x^2u_t}q 
         - \dfrac{(a - 1)f_1u_y}{u_x^2}q_x;\\[3pt]
    w_{2,t}
    &= & -\dfrac{f_2u_{yt}}{u_y^2}q + \dfrac{f_2 u_t}{u_y^2}q_y,\\
    w_{2,x}
    &= & -\dfrac{f_2\big(au_xu_{yt} - (a - 1)u_yu_{xt}\big)}{u_y^2u_t}q
         - \dfrac{(a - 1)f_2}{u_y}q_x + \dfrac{af_2u_x}{u_y^2}q_y;\\[3pt]
    w_{3,y}
    &= & -\dfrac{af_3u_{yt}}{u_t^2}q + \dfrac{f_3}{u_t}q_y
         + \dfrac{(a - 1)f_3u_y}{u_t^2}q_t,\\
    w_{3, x}
    &= & -\dfrac{af_3u_{xt}}{u_t^2}q + \dfrac{af_3u_x}{u_t^2}q_t
  \end{array}
\end{equation}

Let us solve the equation~$\tilde{\ell}_{\mathcal{E}}(\rho) = 0$, where
\begin{multline*}
  \rho = \rho_{0} q + \rho_{1} q_x + \dots + \rho_{15} q_{ttt}
  + \rho_{16} w_1 + \dots + \rho_{19}
  w_{1,xxx}\\ + \rho_{20} w_2 + \dots +
  \rho_{23} w_{2,yyy}
  + \rho_{24} w_3 + \dots + \rho_{27}
  w_{3,ttt},
\end{multline*}
and the coefficients~$\rho_i$ are arbitrary smooth functions on~$\mathcal{E}$
of jet order~$3$. Then, by straightforward computations, we obtain the
following nontrivial solutions
\begin{equation}
  \label{eq:a8}
  \mathcal{R}_1 = \frac{u_xw_1}{f_1},\quad
  \mathcal{R}_2 = \frac{u_yw_2}{f_2},\quad
  \mathcal{R}_3 = \frac{u_tw_3}{f_3}.
\end{equation}
Rewrite System~\eqref{eq:8} in the form
\begin{equation}
  \label{eq:a9}
  \begin{array}{rcl}
    D_t(\phi_1)
    &= & \dfrac{f_1u_t}{u_x^2}D_x(\phi) - \dfrac{af_1}{u_x}D_t(\phi)
         + \dfrac{(a - 1)f_1u_{xt}}{u_x^2}\phi,\\
    D_y(\phi_1)
    &= & -\dfrac{(a - 1)f_1\big((a - 1)u_yu_{xt} -
         au_xu_{yt}\big)}{u_x^2u_t}\phi\\
         
    &-&\dfrac{(a - 1)f_1u_y}{u_x^2}D_x(\phi);\\[16pt]
    D_t(\phi_2)
    &= & -\dfrac{f_2u_{yt}}{u_y^2}\phi + \dfrac{f_2 u_t}{u_y^2}D_y(\phi),\\
    D_x(\phi_2)
    &= & -\dfrac{f_2\big(au_xu_{yt} - (a - 1)u_yu_{xt}\big)}{u_y^2u_t}\phi
         - \dfrac{(a - 1)f_2}{u_y}D_x(\phi)\\         
    &+&\dfrac{af_2u_x}{u_y^2}D_y(\phi);\\[16pt]
    D_y(\phi_3)
    &= & -\dfrac{af_3u_{yt}}{u_t^2}\phi + \dfrac{f_3}{u_t}D_y(\phi)
         + \dfrac{(a - 1)f_3u_y}{u_t^2}D_t(\phi),\\
    D_x(\phi_3)
    &= & -\dfrac{af_3u_{xt}}{u_t^2}\phi + \dfrac{af_3u_x}{u_t^2}D_t(\phi).
  \end{array}
\end{equation}

\begin{proposition}
  Let $\phi$ be a symmetry of Equation~\eqref{eq:6} with $a + b + 1 = 0$
  and~$\phi_1$\textup{,} $\phi_2$\textup{,} $\phi_3$ be solutions to
  System~\eqref{eq:a9}. Then
  \begin{equation*}
    \phi_1 = \mathcal{R}_1(\phi) = \frac{u_x\phi}{f_1},\quad
    \phi_2 = \mathcal{R}_2(\phi) = \frac{u_y\phi}{f_2},\quad
    \phi_3 = \mathcal{R}_3(\phi) = \frac{u_t\phi}{f_3}.    
  \end{equation*}
  are symmetries of~\eqref{eq:6} as well\textup{,} i.e.\textup{,} system of
  relations~\eqref{eq:a8} and~\eqref{eq:a9} is a B\"{a}cklund transformation
  that defines a recursion operator for symmetries of~\eqref{eq:6}. The
  inverse statement is also valid.
\end{proposition}

\subsubsection{The cotangent covering}
\label{sec:cotangent-covering}

Two-component conservation laws from Appendix~\ref{sec:cons-laws-cotang} lead
to the nonlocal variables
\begin{align*}
  w_{1,y} &= \frac{(a - 1)f_1((2 a - 1)u_yu_{xt} - 2au_xu_{yt})u_x}{u_t^2}p
            - \frac{(a - 1)f_1u_y u_x}{u_t}p_x,\\
  w_{1,t} &= \frac{f_1(au_xu_{tt} - (2a - 1)u_tu_{xt})u_x}{u_t^2}p
            + f_1u_xp_x - \frac{af_1u_x^2}{u_t}p_t;\\[8pt]
  w_{2,x} &= \frac{f_2u_y(au_xu_{yt} -(a - 1)u_yu_{xt})}{u_t^2}p
            - \frac{(a - 1)f_2u_y^2}{u_t}p_x + \frac{af_2u_yu_x}{u_t}p_y,\\
  w_{2,t}&= \frac{f_2u_yu_{yt}}{u_t}p + f_2u_yp_y;\\[8pt]
  w_{3,x} &= -\frac{f_3(au_xu_{tt} - (2a - 1)u_tu_{xt})}{u_t}p
            + af_3u_xp_t - f_3u_tp_x;\\
  w_{3,y} &= \frac{(a - 1)f_3(2u_tu_yt - u_yu_{tt})}{u_t}p + (a - 1)f_3u_yp_tu_y.
\end{align*}
Solving the equation $\tilde{\ell}_{\mathcal{E}}(\rho) = 0$ in this covering
with respect to the function\begin{multline*}
  \rho = \rho_{0} p + \rho_{1} p_x + \dots + \rho_{15} p_{ttt}
  + \rho_{16} w_1 + \dots + \rho_{19}w_{1,xxx}\\
  + \rho_{20} w_2 + \dots +
  \rho_{23} w_{2,yyy}
  + \rho_{24} w_3 + \dots + \rho_{27}
  w_{3,ttt},
\end{multline*}
we get the following nontrivial solutions
\begin{equation*}
  \mathcal{R}_1 = \frac{u_tw_1}{f_1u_x^2},\quad
  \mathcal{R}_2 = \frac{u_tw_2}{f_2u_y^2},\quad
  \mathcal{R}_3 = \frac{w_3}{f_3u_t}.
\end{equation*}
that provide recursion operators for cosymmetries.

In all the subsequent cases, nontrivial structures arise as solutions of
$\tilde{\ell}_{\mathcal{E}}(\rho) = 0$ in the coverings
over~$\mathcal{T}\mathcal{E}$ only, i.e., the corresponding B\"{a}cklund
transformations relate symmetries and cosymmetries\footnote{Other solutions
  will be presented in the second part.}.

\subsection{$A = -C$, $B = -C$}
\label{sec:=-c-b}

The structure reads
\begin{multline*}
  \mathcal{I} = \frac{u_t(2f_2u_y - f_{2, y}u)}{f_2u_y}q
  - \frac{u_t(2f_1u_x - f_{1, x}u)}{f_1^2u_x}w_{11} -
  \frac{uu_t}{f_1u_x}w_{11,x} 
  - \frac{2u_t}{f_1u_x}w_{12}\\
  - \frac{u_t(2f_2u_y - f_{2, y}u)}{f_2^2u_y}w_{21} -
  \frac{uu_t}{f_2u_y}w_{21,y} - \frac{u_t}{f_2u_y}w_{22} - \frac{uu_t}{u_y}q_y, 
\end{multline*}
where the nonlocal variables are defined by
\begin{align*}
  w_{11,y}
  &= \frac{f_1(u_xu_{yt} + u_yu_{xt})}{u_xu_t}q
             - \frac{f_1u_y}{u_x} q_x ;\\
  w_{11,t}
  &= \frac{f_1u_{xt}}{u_x}q - \frac{f_1u_t}{u_x}q_x + f_1q_t,\\
  w_{12,y}
  &= f_1\alpha q - \frac{f_1u_y(uu_{xx} + 2u_x^2)}{2 u_x^2}q_x
    + \frac{f_1uu_y}{2u_x}q_{xx},
    \intertext{where}
    \alpha
  &= -\frac{uu_yu_{xxt}}{2u_xu_t} - \frac{uu_yu_{xtt}}{2u_t^2}
    - \frac{uu_xu_{ytt}}{2u_t^2}
    + \frac{uu_yu_{xt}u_{xx}}{2u_tu_x^2}\\
  &+ \Big(\frac{u_y(uu_{tt} + 2u_t^2)}{2u_t^3}
    - \frac{uu_{yt}}{u_t^2}\Big)u_{xt}
    + \frac{u_x(uu_{tt} + 2u_t^2) u_yt}{2u_t^3},\\
  w_{12,t}
  &= \frac{f_1((uu_xx + u_x^2)u_{xt} - uu_xu_{xxt})}{2u_x^2}q
    - \frac{f_1u_t(uu_{xx} + u_x^2)}{2u_x^2}q_x\\
  &+ \frac{f_1u_x}{2}q_t + \frac{f_1uu_t}{2u_x}q_xx
    - \frac{f_1u}{2}q_{xt};\\
  w_{21,x}
  &= \frac{f_2(u_xu_{yt} + u_yu_{xt})}{u_tu_y }q- f_2q_x
    - \frac{f_2 u_x}{u_y}q_y,\\
  w_{21,t}
  &= \frac{f_2u_{yt}}{u_y}q - \frac{f_2u_t}{u_y}q_y;\\
  w_{22,x}
  &= \beta q - (f_{2, y}u + 2f_2u_y) q_x
    - \frac{f_2u_x(uu_{yy} + 2 u_y^2)}{u_y^2}q_y
    + \frac{f_2uu_x}{u_y}q_{yy},
    \intertext{where}
    \beta
  &= -\frac{f_2uu_yu_{xtt}}{u_t^2} - \frac{f_2uu_xu_{yyt}}{u_tu_y}
    - \frac{f_2uu_xu_{ytt}}{u_t^2}
    - \frac{f_2u(2u_tu_{yt} - u_yu_{tt}) u_xt}{u_t^3}\\
    &+ \frac{f_2uu_x(u_t^2u_{yy} + u_y^2u_tt) u_yt}{u_y^2u_t^3} - f_{2, y}u_x;\\ 
  w_{22,t}
  &= -\frac{f_{2,y}u_t u_y^2 + f_2(uu_yu_{yyt} + (u_y^2 -
    uu_{yy})u_{yt})}{u_y^2}q 
    - \frac{f_2u_t(uu_{yy} + u_y^2)}{u_y^2}q_y\\
  &- (f_{2, y}u + f_2u_y)q_t + \frac{f_2uu_tq_{yy}}{u_y} - f_2uq_{yt}.
\end{align*}

\subsection{$A = -C$, $B = C$}
\label{sec:=-c-b-1}

One has
\begin{multline*}
  \mathcal{I} = \left(8u_t + \frac{f_{3, t}u}{f_3} + \frac{f_{2,y}uu_t}{f_2u_y }\right) q
  + \frac{uu_t}{u_y}q_y + u q_t
  - \frac{u_t (2f_2u_y - f_{2, y}u)}{f_2^2u_y}w_{21}\\
  - \frac{u u_t}{f_2u_y}w_{21,y} - \frac{u_t}{f_2u_y}w_{22}
  - \frac{2 u_t f_3 - u f_{3,t}}{f_3^2}w_{31} - \frac{u}{f_3}w_{31,t}
  - \frac{2}{f_3} w_{32},
\end{multline*}
where the defining equations for the nonlocal variables are
\begin{align*}
  w_{21,x}
  &= \frac{f_2(u_xu_{yt} - u_yu_{xt})}{u_yu_t}q + f_2q_x -
    \frac{f_2u_x}{u_y}q_y,\\
  w_{21,t} &= \frac{f_2u_{yt}}{u_y}q - \frac{f_2u_t}{u_y}q_y;\\
  w_{22,x}
  &= \alpha q + (2f_2u_y + f_{2,y} u)q_x - \frac{f_2uu_xu_{yy}}{u_y^2}q_y +
    f_2\frac{u u_x}{u_y}q_{yy},
    \intertext{where}
    \alpha
           &= -\frac{f_2u_yuu_{xtt}}{u_t^2} - \frac{f_2uu_xu_{yyt}}{u_yu_t}
             + \frac{f_2uu_xu_{ytt}}{u_t^2}
             + \frac{f_2u_y(uu_{tt} - 2u_t^2) u_{xt}}{u_t^3}\\
  &+ \frac{f_2uu_xu_{yt}u_{yy}}{u_tu_y^2}
    - \frac{f_2u_x(uu_{tt} - 2u_t^2) u_yt}{u_t^3} + f_{2, y}u_x,\\
  w_{22,t}
  &= \left(\frac{f_2(uu_{yy} + 3u_y^2) u_{yt}}{u_y^2} + u_tf_{2,y}
    - \frac{f_2uu_{yyt}}{u_y}\right) q
    - \frac{f_2u_t(uu_{yy} + u_y^2)}{u_y^2}q_y\\
  &- \frac{(-f_2u_y - f_{2,y}u)q_t + f_2uu_t}{u_y}q_{yy} + f_2uq_{yt};\\
  w_{31,x} &= \frac{f_3u_{xt}}{u_t}q + f_3q_x - \frac{f_3u_x}{u_t}q_t,\\
  w_{31,y} &= \frac{f_3u_{yt}}{u_t}q - \frac{f_3u_y}{u_t}q_t;\\
  w_{32,x} &= \frac{(-f_3uu_tu_{xtt} + f_3(uu_{tt} + 2u_t^2)u_{xt} + f_{3,
             t}u_t^2 u_x)}{2 u_t^2}q + \frac{f_{3,t} u + 2f_3u_t}{2}q_x\\
  &- \frac{f_3uu_xu_{tt}}{2u_t^2}q_t + \frac{f_3uu_x}{2u_t}q_{tt},\\
  w_{32,y}
  &= -\frac{f_3uu_tu_{ytt} - f_3uu_{yt}u_{tt} - u_t^2 (f_{3,t}u_y +
    3f_3u_{yt})}{2 u_t^2}q\\ 
  &+ \frac{f_{3, t}u + f_3u_t}{2}q_y - \frac{f_3uu_yu_{tt} +
    f_3u_t^2u_y}{2u_t^2}q_t 
    +  \frac{f_3uu_y}{2u_t}q_{tt} + \frac{f_3u}{2}q_{yt}.
\end{align*}

\subsection{$A = -C$, $2B = -3C$}
\label{sec:a-=-c}

We have
\begin{equation*}
  \mathcal{I} = -u_tq + \frac{u_t}{f_1\sqrt{u_x}}w_1
  + \frac{8u_t}{9f_2u_y}w_{21}
  - \frac{4(f_{2, y} u - 2 f_2u_y)u_t}{9f_2^2u_y}w_{22}
  + \frac{4uu_t}{9f_2u_y}w_{22,y}, 
\end{equation*}
where
\begin{align*}
  w_{1,y}
  &= \frac{3f_1(2u_xu_{yt} + 3u_yu_{xt})}{4\sqrt{u_x}u_t}q
    - \frac{3f_1 u_y}{2\sqrt{u_x}}q_x,\\
  w_{1,t}
  &= \frac{3f_1u_{xt}}{2\sqrt{u_x}}q - \frac{f_1u_t}{\sqrt{u_x}}q_x
    + f_1\sqrt{u_x}q_t;\\
  w_{21,x}
  &= f_2 \alpha q + \frac{3f_2(uu_{yt} - u_tu_y)}{4u_t}q_x
    - \frac{f_2(4uu_tu_xu_{yy} - 9uu_y^2u_{xt} + 4u_tu_xu_y^2)}{8u_y^2u_t}q_y\\
  &- \frac{3f_2u(2u_xu_{yt} + 3u_yu_{xt})}{8u_t^2}q_t
    + \frac{9f_2uu_y}{8u_t}q_{xt} + \frac{f_2uu_x}{2u_y}q_{yy}
    + \frac{3f_2uu_x}{4u_t}q_{yt},
    \intertext{where}
    \alpha
  &= -\frac{9uu_yu_{xtt}}{8u_t^2} - \frac{uu_xu_{yyt}}{2u_yu_t}
    - \frac{3uu_xu_{ytt}}{4u_t^2}
    - \frac{3(5uu_tu_{yt} - 3uu_yu_{tt} - 2u_t^2u_y)u_{xt}}{8u_t^3}\\
  &+ \frac{uu_xu_{yt}u_{yy}}{2u_y^2u_t}
    + \frac{u_x(3uu_{tt} + 2u_t^2)u_{yt}}{4u_t^3},\\
    w_{21,t}
  &= -\frac{f_2((u_y^2 - 2uu_{yy})u_{yt} + 2uu_yu_{yyt})}{4u_y^2}q\\
    &- \frac{f_2u_t(2uu_{yy} - u_y^2)}{4u_y^2}q_y +
      \frac{f_2uu_tq_{yy}}{2u_y};\\ 
    w_{22,x}
  &= \frac{f_2(2u_xu_{yt} + 3u_yu_{xt})}{2u_yu_t}q
    - \frac{3f_2}{2}q_x - \frac{f_2u_x}{u_y}q_y,\\
  w_{22,t}
  &= \frac{f_2u_{yt}}{u_y} q - \frac{f_2u_t}{u_y}q_y.
\end{align*}

\subsection{$A = -C$, $B = -3C$}
\label{sec:a-=-c-3}

Here we have
\begin{equation*}
  \mathcal{I} = \frac{u_t}{f_2u_y}w_{21}
  + \frac{(2f_2u_y - f_{2, y}u) u_t}{2f_2^2u_y}w_{22}
  + \frac{uu_t}{2f_2u_y}w_{22,y}
\end{equation*}
with the nonlocal variables defined by
\begin{align*}
  w_{21,x}
  &= f_2\alpha q + \frac{3f_2(uu_{yt} - u_tu_y)}{2u_t}q_x
    - \frac{f_2((u_tu_x - 9uu_{xt})u_y^2 + uu_tu_xu_{yy})}{2u_y^2u_t}q_y\\
  &- \frac{3f_2u(u_xu_{yt} + 3u_yu_{xt})}{2u_t^2}q_t
    + \frac{9f_2uu_y}{2u_t}q_{xt} + \frac{f_2uu_x}{2u_y}q_{yy}
    + \frac{3f_2uu_x}{2u_t}q_{yt},
    \intertext{where}
    \alpha
  &= -\frac{9uu_yu_{xtt}}{2u_t^2} - \frac{uu_xu_{yyt}}{2u_yu_t}
    - \frac{3uu_xu_{ytt}}{2u_t^2} - \frac{6uu_{yt}}{u_t^2}
    + \frac{3u_y(3uu_{tt} + u_t^2)u_{xt}}{2u_t^3}\\
  &+ \frac{uu_xu_{yt}u_{yy}}{2u_y^2u_t}
    + \frac{u_x(3uu_{tt} + u_t^2)u_{yt}}{2u_t^3},\\
  w_{21,t}
  &= -\frac{f_2((2u_y^2 - uu_{yy})u_{yt} + uu_yu_{yyt})}{2u_y^2}q
    - \frac{f_2u_t(uu_{yy} - 2 u_y^2)}{2u_y^2}q_y
    + \frac{f_2uu_t}{2u_y}q_{yy};\\
  w_{22,x}
  &= \frac{f_2(u_xu_{y}t + 3u_yu_{xt})}{u_yu_t}q - 3f_2q_x -
    \frac{f_2u_x}{u_y}q_y,\\
  w_{22,t}
  &= \frac{f_2u_yt}{u_y}q - \frac{f_2u_t}{u_y}q_y.
\end{align*}

\subsection{$A = -C$, $2B = 3C$}
\label{sec:a-=-c-2}

We have
\begin{equation*}
  \mathcal{I} = -\frac{27u_t}{8}q + \frac{u_t}{f_2u_y}w_{21}
  - \frac{(f_{2, y}u - 2f_2u_y)u_t}{2f_2^2u_y}w_{22}
  + \frac{u u_t}{2f_2u_y}w_{22,y} + \frac{9\sqrt{u_t}}{8f_3}w3,
\end{equation*}
where
\begin{align*}
  w_{21,x}
  &= f_2\alpha q - \frac{3f_2(uu_{yt} - u_tu_y)}{4u_t}q_x
    - \frac{f_2(4uu_tu_xu_{yy} - 9uu_y^2u_{xt} + 4
    u_tu_xu_y^2)}{8u_y^2u_t}q_y\\
  &+ \frac{3f_2u(2u_xu_{yt} - 3u_yu_{xt})}{8u_t^2}q_t
    + \frac{f_2uu_x}{2u_y}q_{yy} - \frac{3f_2uu_x}{4u_t}q_{yt}
    + \frac{9f_2uu_y}{8u_t}q_{xt},
    \intertext{where}
    \alpha
  &= -\frac{9uu_yu_{xtt}}{8u_t^2} - \frac{uu_xu_{yyt}}{2u_yu_t}
    + \frac{3uu_xu_{ytt}}{4u_t^2} - \frac{3uu_xu_{yt}u_{tt}}{4u_t^3}\\
  &- \frac{3(uu_tu_{yt} - 3uu_yu_{tt} + 2u_t^2u_y)u_{xt}}{8u_t^3}
     + \frac{u_x(uu_{yy} + u_y^2)u_{yt}}{2u_tu_y^2},\\
  w_{21,t}
  &= \frac{f_2((2uu_{yy} + 5u_y^2)u_{yt} - 2uu_yu_{yyt})}{4u_y^2}q
    - \frac{f_2u_t(2uu_{yy} + 5u_y^2)}{4u_y^2}q_y\\
    &+ \frac{f_2uu_t}{2u_y}q_{yy};\\
  w_{22,x}
  &= \frac{f_2(2u_xu_{yt} - 3u_yu_{xt})}{2u_yu_t}q
    + \frac{3f_2}{2}q_x - \frac{f_2u_x}{u_y}q_y,\\
  w_{22,t} &= \frac{f_2u_{yt}}{u_y}q - \frac{f_2 u_t}{u_y}q_y;\\
  w_{3,x}
  &= \frac{f_3(3u_{xt}q  - 2u_xq_t  + 2u_tq_x)}{2\sqrt{u_t}},\\
  w_{3,y} &= \frac{3f_3(u_{yt}q - u_yq_t)}{2\sqrt{u_t}}.
\end{align*}

\subsection{$A = -C$, $B = 3C$}
\label{sec:a-=-c-1}

The structure is
\begin{equation*}
  \mathcal{I} = -\frac{9u_t}{2}q + \frac{u_t}{f_2u_y}w_{21}
  - \frac{(f_{2, y}u - 2f_2u_y)u_t}{2f_2^2u_y}w_{22}
  + \frac{uu_t}{2f_2u_y}w_{22,y},
\end{equation*}
with the nonlocal variables defined by
\begin{align*}
  w_{21,x}
  &= f_2\alpha q - \frac{3f_2(uu_{yt} - u_tu_y)}{2u_t}q_x
    - \frac{f_2((u_tu_x - 9uu_{xt})u_y^2 + uu_tu_xu_{yy})}{2u_y^2u_t}q_y\\
  &+ \frac{3f_2u(u_xu_{yt} - 3u_yu_{xt})}{2u_t^2}q_t
    + \frac{9f_2uu_y}{2u_t}q_{xt} + \frac{f_2uu_x}{2u_y}q_{yy}
    - \frac{3f_2uu_x}{2u_t}q_{yt},
    \intertext{where}
    \alpha
  &= -\frac{9uu_yu_{xtt}}{2u_t^2} - \frac{uu_xu_{yyt}}{2u_yu_t}
    + \frac{3uu_xu_{ytt}}{2u_t^2}
    - \frac{3u(u_xu_{yt} - 3u_yu_{xt})u_{tt}}{2u_t^3}\\
  &+ \frac{uu_xu_{yt}u_{yy}}{2u_y^2u_t}
    - \frac{3(2uu_{yt} + u_tu_y)u_{xt}}{2u_t^2} + \frac{u_xu_{yt}}{2u_t},\\
  w_{21,t}
  &= -\frac{f_2(uu_yu_{yyt} - u u_{yt}u_{yy} - 4u_y^2u_{yt})}{2u_y^2}q
    - \frac{f_2u_t(uu_{yy} + 4u_y^2)}{2u_y^2}q_y\\
  &+ \frac{f_2uu_t}{2u_y}q_{yy};\\
  w_{22,x}
  &= \frac{f_2 (u_xu_{yt} - 3u_yu_{xt})}{u_yu_t}q
    + 3f_2q_x - \frac{f_2u_x}{u_y}q_y,\\
    w_{22,t}
  &= \frac{f_2u_{yt}}{u_y}q - \frac{f_2u_t}{u_y}q_y.
\end{align*}

\subsection{$3A = 2C$, $3B = -2C$}
\label{sec:3a-=-2c}

The structure is
\begin{equation*}
  \mathcal{I} = \frac{f_{3,t}u + f_3u_t}{f_3}q + uq_t
  + \frac{9u_t}{4f_1\sqrt{u_x}}w_1
  + \frac{f_{3,t}u - 2f_3u_t}{f_3^2}w_{31} - \frac{u}{f_3}w_{31,t} -
  \frac{2}{f_3}w_{32}
\end{equation*}
with
\begin{align*}
  w_{1,y}
  &= -\frac{4f_1(u_xu_{yt} - u_yu_{xt})}{9\sqrt{u_x}u_t}q
    - \frac{2f_1u_y}{3\sqrt{u_x}}q_x,\\
  w_{1,t}
  &= \frac{2f_1u_{xt}}{3\sqrt{u_x}}q - \frac{f_1u_t}{\sqrt{u_x}}q_x
    - \frac{2f_1\sqrt{u_x}}{3}q_t;\\
  w_{31,x}
  &= -\frac{2f_3u_{xt}}{3u_t}q + f_3q_x + \frac{2f_3u_x}{3u_t}q_t,\\
  w_{31,y}
  &= -\frac{2f_3u_{yt}}{3u_t}q + \frac{2f_3u_y}{3u_t}q_t;\\
  w_{32,x}
  &= \frac{(3f_{3,t}u_x  + f_3u_{xt}) u_t^2 + 2f_3uu_tu_{xtt} -
    2f_3uu_{xt}u_{tt}}{6u_t^2}q 
    + \frac{f_3u_x(2uu_{tt} + 5 u_t^2)}{6u_t^2}q_t\\
  &+ \frac{f_{3,t}u + 2f_3u_t}{2}q_x - \frac{f_3uu_x}{(3 u_t)}q_{tt},\\
  w_{32,y}
  &= \frac{(3f_{3,t}u_y + 4f_3u_{yt})u_t^2 + 2f_3uu_tu_{ytt} -
    2f_3uu_{yt}u_{tt}}{6u_t^2}q + \frac{f_{3, t}u + f_3u_t}{2}q_y\\
  &+ \frac{f_3u_y(uu_{tt} + u_t^2)}{3u_t^2}q_t
    - \frac{f_3uu_y}{3u_t}q_{tt} + \frac{f_3u}{2}q_{yt}.
\end{align*}

\subsection{$3A = C$, $3B = -C$}
\label{sec:3a-=-c}

The structure reads
\begin{equation*}
  \mathcal{I} = -\frac{f_3u_t + f_{3, t}u}{2f_3}q - \frac{u}{2}q_t
  - \frac{f_{3,t}u - 2f_3u_t}{2f_3^2}w_{31} + \frac{u}{2f_3}w_{31,t}
  + \frac{1}{f_3}w_{32}
\end{equation*}
with
\begin{align*}
  w_{31,x}
  &= -\frac{f_3u_{xt}}{3u_t}q + f_3q_x + \frac{f_3u_x}{3u_t}q_t,\\
  w_{31,y}
  &= -\frac{f_3u_{yt}}{3u_t}q + \frac{f_3u_y}{3u_t}q_t;\\
  w_{32,x}
  &= \frac{f_3(2u_t^2 - uu_{tt})u_{xt} + 3f_{3,t}u_t^2u_x + f_3uu_tu_{xtt}}{6
    u_t^2}q 
    + \frac{f_3u_x(uu_{tt} + 4u_t^2)}{6u_t^2}q_t\\
  &+ \frac{f_{3,t} u + 2f_3u_t}{2}q_x - \frac{f_3uu_x}{6u_t}q_{tt},\\
  w_{32,y}
  &= \frac{3f_{3,t}u_t^2u_y + f_3uu_tu_{ytt} - f_3uu_{yt}u_{tt} +
    5f_3u_t^2u_{yt}}{6u_t^2}q\\
  &+ \frac{f_3u_t + f_{3,t}u}{2}q_y + \frac{f_3u_y(uu_{tt} +
    u_t^2)}{6u_t^2}q_t 
    + \frac{f_3u}{2}q_{yt} - \frac{f_3uu_y}{6u_t}q_{tt}.
\end{align*}

\subsection{$A = C$, $2B = C$}
\label{sec:a-=-c-4}

The structure is
\begin{equation*}
  \mathcal{I} = u_tq - \frac{u_t}{f_1\sqrt{u_x}}w_1 - \frac{\sqrt{u_t}}{f_3}w_3
\end{equation*}
with the nonlocalities defined by
\begin{align*}
  w_{1,y}
  &= -\frac{f_1(2u_xu_{yt} - u_yu_{xt})}{4\sqrt{u_x}u_t}q
    - \frac{f_1u_y}{2\sqrt{u_x}}q_x,\\
  w_{1,t}
  &= f_1\left(\frac{u_{xt}}{2\sqrt{u_x}}q -  \frac{u_t}{\sqrt{u_x}}q_x
    - \sqrt{u_x}q_t\right);\\
  w_{3,x}
  &= \frac{f_3(u_{xt}q - u_xq_t)}{\sqrt{u_t}},\\
  w_{3,y}
  &= f_3\left(\frac{u_{yt}}{\sqrt{u_t}}q - \frac{u_y}{2\sqrt{u_t}}q_t
    + \sqrt{u_t}q_y\right).
\end{align*}

\subsection{$A = -2C$, $B = -2C$}
\label{sec:a-=-2c}

The structure:
\begin{equation*}
  \mathcal{I} = u_tq - \frac{u_t}{2f_1\sqrt{u_x}}w_1
  - \frac{u_t}{2f_2\sqrt{u_y}}w_2,
\end{equation*}
defining equations:
\begin{align*}
  w_{1,y}
  &= \frac{4f_1(u_xu_{yt} + u_yu_{xt})}{\sqrt{u_x}u_t}q
    - \frac{2f_1u_y}{\sqrt{u_x}}q_x,\\
  w_{1,t}
  &= \frac{2f_1u_{xt}}{\sqrt{u_x}}q - \frac{f_1u_t}{\sqrt{u_x}}q_x
    + 2f_1\sqrt{u_x}q_t;\\
  w_{2,x}
  &= \frac{2f_2(u_xu_{yt} + u_yu_{xt})}{\sqrt{u_y}u_t}q
    - 2f_2\sqrt{u_y}q_x - \frac{2f_2u_x}{\sqrt{u_y}}q_y,\\
  w_{2,t}
  &= \frac{f_2u_{yt}}{\sqrt{u_y}}q - \frac{f_2u_t}{\sqrt{u_y}}q_y.
\end{align*}

\section*{Conclusion}
\label{sec:conclusion}

In the forthcoming second part of this study, we shall use the above results
to construct infinite hierarchies of integrable structures, nonlocal
symmetries, and cosymmetries.

\section*{Acknowledgements}
\label{sec:acknowledgements}

Computations were done with the \textsc{Jets} software,~\cite{BM-Jets}.

\appendix

\section{Cosymmetries in the case $A + B + C = 0$}
\label{sec:cosymmetries-case-a}

The cosymmetries up to jet order $3$ are
\begin{align*}
  \psi_0 &= f_0 u_t,\\
  \psi_1 &= \frac{f_1u_t}{u_x^2},\\
  \psi_2 &= \frac{f_2 u_t}{u_y^2},\\
  \psi_3&= \frac{f_3}{u_t},\\
  \psi_{41} &= \frac{u_{ytt}}{u_yu_t} - \frac{u_{yt}^2}{2u_tu_y^2}
              - \frac{u_{tt}u_{yt}}{u_t^2u_y},\\
  \psi_{42} &= \frac{u_{yyt}}{u_y^2} - \frac{u_{yt}^2}{2u_tu_y^2}
              - \frac{u_{yy}u_{yt}}{u_y^3},\\
  \psi_{43} &= \frac{u_{xxt}}{u_x^2} - \frac{u_{xt}^2}{2u_tu_x^2}
              - \frac{u_{xt}u_{xx}}{u_x^3},\\
  \psi_{44} &= \frac{u_{xtt}}{u_xu_t} - \frac{u_{xt}^2}{2u_tu_x^2}
              - \frac{u_{tt}u_{xt}}{u_t^2u_x},\\
  \psi_{45} &= \frac{u_tu_{yyy}}{u_y^3} - \frac{3u_tu_{yy}^2}{2u_y^4},\\
  \psi_{46} &= \frac{u_tu_{xxx}}{u_x^3} - \frac{3u_tu_{xx}^2}{2u_x^4},\\
  \psi_{47} &= \frac{u_tu_{xxx}}{u_x^3} + \frac{u_{ttt}}{u_t^2}
              - \frac{3u_{tt}^2}{2u_t^3} - \frac{3u_tu_{xx}^2}{2u_x^4},
\end{align*}
where $f_0 = f_0(u)$, $f_1 = f_1(x)$, $f_2 = f_2(y)$, $f_3 = f_3(t)$.

\section{Conservation laws of the tangent equation in the case $A + B + C =
  0$} 
\label{sec:cons-laws-tang}

The conservation laws of $\mathcal{T}\mathcal{E}$
\begin{equation*}
  \omega = X\,dy\wedge dt + Y\,dt\wedge dx + T\,dx\wedge dy
\end{equation*}
that correspond to the cosymmetries above are
\begin{align*}
  X_0 &= \big(f_0(u_{yt} - 2au_{yt}) - (a - 1)f_{0,u}u_yu_t\big)q + f_0u_tq_y,\\
  Y_0 &= (a - 1)(f_{0, u}u_tu_x  + 2f_0u_{xt}) q,\\
  T_0 &= f_0 \frac{(1 - a)u_yu_{xt} + au_xu_{yt}}{u_t} q - af_0u_xq_y
        + (a - 1)f_0u_yq_x;\\[3pt] 
  X_1 &= 0,\\
  Y_1 &= \frac{f_1u_t}{u_x^2}q_x - \frac{af_1}{u_x}q_t
        + \frac{(a - 1)f_1u_{xt}}{u_x^2}q,\\
  T_1 &= \frac{(a - 1)f_1\big((a - 1)u_yu_{xt} - au_xu_{yt}\big)}{u_x^2u_t}q
        + \frac{(a - 1)f_1u_y}{u_x^2}q_x;\\[3pt]
  X_2 &= -\frac{f_2u_{yt}}{u_y^2}q + \frac{f_2 u_t}{u_y^2}q_y,\\
  Y_2&= 0,\\
  T_2 &= \frac{f_2\big(au_xu_{yt} - (a - 1)u_yu_{xt}\big)}{u_y^2u_t}q
        + \frac{(a - 1)f_2}{u_y}q_x - \frac{af_2u_x}{u_y^2}q_y;\\[3pt]
  X_3 &= -\frac{af_3u_{yt}}{u_t^2}q + \frac{f_3}{u_t}q_y
        + \frac{(a - 1)f_3u_y}{u_t^2}q_t,\\
  Y_3 &= \frac{af_3u_{xt}}{u_t^2}q - \frac{af_3u_x}{u_t^2}q_t.\\
  T_3 &= 0.
\end{align*}
The expressions for the rest of conservation laws are so big that it makes no
sense to present them here. The only thing that is to be mentioned is that all
them are essentially $3$-component ones.

\section{Conservation laws of the cotangent equation in the case $A + B + C =
  0$}
\label{sec:cons-laws-cotang}

These laws are
\begin{align*}
  X_0 &= \frac{f_{0,u} (a - 1) u_y u_t - a f_0u_{yt}}{u_t} p - f_0p_y,\\
  Y_0 &= \frac{(a - 1)(f_0u_{xt} - f_{0, u}u_xu_t)}{u_t} p,\\
  T_0 &= \frac{f_0(au_xu_{yt} - (a - 1)u_yu_{xt})}{u_t^2}p
        - \frac{(a - 1)f_0u_y}{u_t}p_x + \frac{af_0u_x}{u_t}p_y;\\
  X_1 &= 0,\\
  Y_1 &= -\frac{f_1(au_xu_{tt} - (2a - 1)u_tu_{xt}) u_x}{u_t^2}p
        - f_1u_xp_x + \frac{af_1u_x^2}{u_t}p_t,\\
  T_1 &= \frac{f_1(a - 1)((2a  - 1)u_yu_{xt} - 2au_xu_{yt})u_x}{u_t^2}p
        - \frac{(a - 1)f_1u_yu_x}{u_t}p_x;\\
  X_2 &= -\frac{f_2u_yu_{yt}}{u_t}p - f_2u_yp_y,\\
  Y_2 &= 0,\\
  T_2 &= \frac{f_2u_y(au_xu_{yt} - (a - 1)u_yu_{xt})}{u_t^2}p
        - \frac{(a - 1)f_2u_y^2}{u_t}p_x +\frac{ af_2 u_yu_x}{u_t}p_y;\\
  X_3 &= \frac{(a - 1)f_3(u_yu_{tt} - 2u_tu_{yt})}{u_t}p - (a - 1)f_3u_yp_t,\\
  Y_3 &= \frac{f_3((2a - 1)u_tu_{xt} - au_xu_{tt})}{u_t}p + af_3u_xp_t
        - f_3u_tp_x,\\
  T_3 &= 0.
\end{align*}

\section{Cosymmetries of $\mathcal{E}$ and conservation laws of $\mathcal{T}\mathcal{E}$ for the cases $A + B +
  C\neq 0$}
\label{sec:cons-laws-mathc}

\subsection{$A =  C$, $B =  C$}
\label{sec:a-=-c-5}

Cosymmetries of $\mathcal{E}$:
\begin{gather*}
  \psi_{11} = \frac{f_{1, x}uu_t}{2u_x} + f_1 u_t,\qquad
\psi_{12} = \frac{f_1u_t}{u_x},\\
\psi_{21} = \frac{f_2 u_t}{u_y},\qquad
\psi_{22} = 2 f_2 u_t + \frac{f_{2, y} u u_t}{u_y},\\
\psi_3 = \frac{f_3}{u_t^2}.
\end{gather*}
Conservation law of $\mathcal{T}\mathcal{E}$:
\begin{align*}
  X_{11} &= 0,\\
  Y_{11} &= -\frac{f_1u_{xt}}{u_x}q + \frac{f_1u_t}{u_x}q_x - f_1q_t,\\
  T_{11} &= \frac{f_1(u_xu_{yt} + u_yu_{xt})}{u_xu_t}q -
           \frac{f_1u_y}{u_x}q_x;\\ 
  X_{12} &= 0,\\
  Y_{12} &= \frac{f_1(uu_xu_{xxt} - (u u_xx + u_x^2)u_{xt})}{2 u_x^2}q
           + \frac{f_1u_t(uu_{xx} + u_x^2)}{2u_x^2} q_x
           - \frac{f_1u_x}{2}q_t\\
         &+ \frac{f_1u}{2}q_{xt} - u u_t  \frac{f_1}{2u_x}q_{xx},\\
  T_{12} &= f_1\alpha q - \frac{f_1u_y(uu_{xx} + 2u_x^2)}{2u_x^2}q_x
           + \frac{f_1uu_y}{(2 u_x)} q_{xx},
           \intertext{where}
           \alpha & = -\frac{uu_yu_{xxt}}{2u_xu_t}
                    - \frac{uu_yu_{xtt}}{2u_t^2} - \frac{uu_xu_{ytt}}{2u_t^2}
                    + \frac{uu_yu_{xt}u_{xx}}{2u_tu_x^2}\\
         &+ \frac{(u_y(uu_{tt} + 2u_t^2) - 2uu_tu_{yt})u_{xt}}{2u_t^3}
           + \frac{u_x(uu_{tt} + 2u_t^2)u_{yt}}{2u_t^3};\\
  X_{21} &= -\frac{f_2 u_{yt}}{u_y}q + \frac{f_2  u_t}{u_y}q_y,\\
  Y_{21} &= 0,\\
  T_{21} &= \frac{f_2 (u_xu_{yt} + u_yu_{xt})}{u_tu_y}q - f_2 q_x -
           \frac{f_2 u_x}{u_y}q_y;\\
  X_{22} &= \big(\frac{f_2uu_{yyt}}{u_y} - \frac{f_2uu_{yt}u_{yy}}{u_y^2}
           + f_2u_{yt} + f_{2, y}u_tu_y\big)q
           + \frac{f_2u_t(uu_{yy} + u_y^2)}{u_y^2}q_y \\
         &+ (f_{2, y}u + f_2u_y)q_t - \frac{f_2uu_t}{u_y}q_{yy} + f_2uq_{yt};\\
  Y_{22} &= 0,\\
  T_{22} &= \beta q - (f_{2, y}u + 2f_2u_y) q_x
           - \frac{f_2 u_x (uu_{yy} + 2u_y^2)}{u_y^2}q_y +
           \frac{f_2uu_x}{u_y}q_{yy}, 
           \intertext{where}
           \beta & = -\frac{f_2uu_yu_{xtt}}{u_t^2} -
                   \frac{f_2uu_xu_{ytt}}{u_t^2}
                   - \frac{f_2uu_xu_{yyt}}{u_t u_y}
                   - \frac{f_2u(2u_tu_{yt} - u_yu_{tt})u_{xt}}{u_t^3}\\
         &+ \frac{f_2uu_xu_{yt}u_{yy}}{u_tu_y^2}
           +\frac{f_2uu_xu_{yt}u_{tt}}{u_t^3} - f_{2, y}u_x;\\
  X_3 &= \frac{f_3u_{yt}}{u_t^3}q - \frac{f_3u_y}{u_t^3}q_t,\\
  Y_3 &= -\frac{f_3u_{xt}}{u_t^3}q + \frac{f_3}{u_t^2}q_x -
        \frac{f_3u_x}{u_t^3}q_t,\\ 
  T_3 &= 0.
\end{align*}

\subsection{$A = -C$, $B = C$}
\label{sec:a-=-c-6}

Cosymmetries:
\begin{gather*}
  \psi_1 = \frac{f_1u_t}{u_x^3};\\
  \psi_{21} = \frac{f_2 u_t}{u_y},\quad
  \psi_{22} = 2f_2u_t + \frac{f_{2, y}u u_t}{u_y};\\
  \psi_{31} = f_3,\quad
  \psi_{32} = \frac{f_{3, t} u}{2} + f_3 u_t.
\end{gather*}
Conservation laws:
\begin{align*}
  X_1
  &= 0,\\
  Y_1
  &= \frac{f_1u_{xt}}{u_x^3}q + \frac{f_1u_t}{u_x^3}q_x -
    \frac{f_1}{u_x^2}q_t,\\ 
  T_1
  &= \frac{f_1(u_yu_{xt} - u_xu_{yt})}{u_x^3u_t}q + \frac{f_1u_y}{u_x^3}q_x;\\
  X_{21}
  &= -\frac{f_2u_{yt}}{u_y}q + \frac{f_2u_t}{u_y}q_y,\\
  Y_{21}
  &= 0,\\
  T_{21}
  &= \frac{f_2(u_xu_{yt} - u_yu_{xt})}{u_yu_t}q + f_2q_x -
    \frac{f_2u_x}{u_y}q_y;\\ 
  X_{22}
  &= \frac{-f_{2, y}u_tu_y^2 + f_2(uu_yu_{yyt} - uu_{yt}u_{yy} -
    3u_y^2u_{yt})}{u_y^2}q 
    + \frac{f_2u_t(uu_{yy} + u_y^2)}{u_y^2}q_y\\
  &- (f_2u_y + f_{2, y}u)q_t - \frac{f_2uu_t}{u_y}q_{yy} - f_2uq_{yt},\\
  Y_{22}
  &= 0,\\
  T_{22}
  &= \alpha q + (2f_2u_y + f_{2, y}u)q_x - \frac{f_2uu_xu_{yy}}{u_y^2}q_y
    + \frac{f_2u u_x}{u_y} q_{yy},
    \intertext{where}
    \alpha
  &= -\frac{f_2uu_yu_{xtt}}{u_t^2} - \frac{f_2uu_xu_{yyt}}{u_yu_t}
    + \frac{f_2uu_xu_{ytt}}{u_t^2}
    + \frac{f_2u_y(uu_{tt} - 2u_t^2) u_{xt}}{u_t^3}\\
  &+ \frac{f_2u_x(uu_t^2u_{yy} - uu_y^2u_{tt} +
    2u_t^2u_y^2)u_{yt}}{u_t^3u_y^2} 
    + f_{2, y}u_x;\\
  X_{31}
  &= -\frac{f_3u_{yt}}{u_t}q + \frac{f_3u_t}{u_y}q_t,\\
  Y_{31}
  &= \frac{f_3u_{xt}}{u_t}q + f_3q_x - \frac{f_3u_x}{u_t}q_t,\\
  T_{31}
  &= 0;\\
  X_{32}
  &= \frac{f_3uu_tu_{ytt} - f_3uu_{yt}u_{tt} - u_t^2(f_{3,t}u_y +
    3f_3u_{yt})}{2u_t^2}q  - \frac{f_{3, t}u + f_3u_t}{2}q_y\\
  &+ \frac{f_3uu_yu_{tt} + f_3u_t^2u_y}{2u_t^2}q_t
    - \frac{f_3 u}{2}q_{yt} -  \frac{f_3 uu_y}{2u_t}q_{tt},\\
  Y_{32}
  &= \frac{-f_3uu_tu_{xtt} + f_3uu_{xt}u_{tt} + u_t^2(f_{3, t}u_x +
    2f_3u_{xt})}{2u_t^2}q \\
  &+ \frac{f_{3, t}u + 2f_3u_t}{2}q_x - \frac{f_3uu_xu_{tt}}{2u_t^2}q_t +
    \frac{f_3uu_x}{2u_t}q_{tt},\\
  T_{32}
  &= 0.
\end{align*}

\subsection{$A = -C$, $2B = -3C$}
\label{sec:a-=-c-7}

Cosymmetries:
\begin{gather*}
  \psi_1 = \frac{f_1 u_t}{\sqrt{u_x}};\\
\psi_{21} = f_2 u_t + \frac{f_{2, y}uu_t}{2u_y},\quad
\psi_{22} = \frac{f_2 u_t}{u_y};\\
\psi_3 = f_3u_t^{-{5}/{2}}.
\end{gather*}
Conservation laws:
\begin{align*}
  X_1
  &= 0,\\
  Y_1
  &= -\frac{3f_1u_{xt}}{2\sqrt{u_x}}q + \frac{f_1 u_t}{\sqrt{u_x}} q_x
    - f_1 \sqrt{u_x} q_t,\\
  T_1
  &= \frac{3f_1(2u_xu_{yt} + 3u_yu_{xt})}{4\sqrt{u_x}u_t}q
    - \frac{3f_1u_y}{2\sqrt{u_x}}q_x;\\
  X_{21}
  &= \frac{f_2(2uu_yu_{yyt} - 2uu_{yt}u_{yy} + u_y^2u_{yt})}{4u_y^2}q\\
  &+ \frac{f_2u_t(2uu_{yy} - u_y^2)}{4u_y^2}q_y
    - \frac{f_2uu_t}{2u_y}q_{yy},\\
  Y_{21}
  &= 0,\\
  T_{21}
  &= f_2\alpha q + \frac{3f_2(uu_{yt} - u_tu_y)}{4u_t}q_x
    - \frac{f_2((4u_tu_x - 9uu_{xt})u_y^2 + 4uu_tu_xu_{yy})}{8u_y^2u_t}q_y\\
  &- \frac{3f_2u(2u_xu_{yt} + 3u_yu_{xt})}{8u_t^2}q_t
    + \frac{9f_2uu_y}{8u_t}q_{xt} + \frac{f_2uu_x}{2u_y}q_{yy}
    + \frac{3f_2uu_x}{4u_t}q_{yt},
    \intertext{where}
    \alpha
  &= -\frac{9uu_yu_{xtt}}{8u_t^2} - \frac{uu_xu_{yyt}}{2u_yu_t}
    - \frac{3uu_xu_{ytt}}{4u_t^2}
    + \frac{3u(2u_xu_{yt} + 3u_yu_{xt})u_{tt}}{8u_t^3}\\
  &+ \frac{uu_xu_{yt}u_{yy}}{2u_y^2u_t}
    - \frac{3(5uu_{yt} - 2u_tu_y)u_{xt}}{8u_t^2} + \frac{u_xu_{yt}}{2u_t};\\
  X_{22}
  &= -\frac{f_2u_{yt}}{u_y}q + \frac{f_2u_t}{u_y}q_y,\\
  Y_{22}
  &= 0,\\
  T_{22}
  &= \frac{f_2(2u_xu_{yt} + 3u_yu_{xt}}{2u_yu_t}q
    - \frac{3f_2}{2}q_x - \frac{f_2u_x}{u_y}q_y;\\
  X_3
  &= \frac{3f_3(u_{yt}q - u_yq_t)}{2u_t^{\frac{7}{2}}},\\
  Y_3
  &= -\frac{f_3 (3u_{xt}q + 2u_xq_t - 2u_tq_x)}{2u_t^{\frac{7}{2}}},\\
  T_3 &= 0.
\end{align*}

\subsection{$A = -C$, $B = -3C$}
\label{sec:a-=-c-8}

Cosymmetries:
\begin{gather*}
  \psi_1 = f_1 u_t u_x;\\
  \psi_{21} = f_2 u_t + \frac{f_{2, y} u u_t}{2u_y},\quad
  \psi_{22} = \frac{f_2 u_t}{u_y};\\
  \psi_3 = \frac{f_3}{u_t^4}.
\end{gather*}
Conservation laws:
\begin{align*}
  X_1
  &= 0,\\
  Y_1
  &= -3f_1u_xu_{xt}q - f_1u_x^2q_t + f_1u_tu_xq_x,\\
  T_1
  &= \frac{3f_1u_x(u_xu_{yt} + 3u_yu_{xt})}{u_t}q - 3f_1u_xu_yq_x;\\
  X_{21}
  &= \frac{f_2((2u_y^2 - uu_{yy}) u_{yt} + uu_yu_{yyt})}{2u_y^2}q
    + \frac{f_2u_t(uu_{yy} - 2u_y^2)}{2u_y^2}q_y
    - \frac{f_2uu_t}{2u_y}q_{yy},\\
  Y_{21}
  &= 0,\\
  T_{21}
  &= f_2 \alpha q + \frac{3f_2(uu_{yt} - u_tu_y)}{2u_t}q_x
    - \frac{f_2((u_tu_x - 9uu_{xt})u_y^2 + uu_tu_xu_{yy})}{2u_y^2u_t}q_y\\
  &- \frac{3f_2u(u_xu_{yt} + 3u_yu_{xt})}{2u_t^2}q_t
    + \frac{9f_2uu_y}{2u_t}q_{xt} + \frac{f_2uu_x}{2u_y}q_{yy}
    + \frac{3f_2uu_x}{2u_t}q_{yt},
    \intertext{where}
    \alpha
  &= -\frac{9uu_yu_{xtt}}{2u_t^2} - \frac{uu_xu_{yyt}}{2u_yu_t}
    - \frac{3uu_xu_{ytt}}{2u_t^2}
    + \frac{3u(u_xu_{yt} + 3u_yu_{xt}) u_{tt}}{2u_t^3}\\
  &+ \frac{uu_xu_{yt}u_{yy}}{2u_y^2u_t}
    - \frac{3(4uu_{yt} - u_tu_y)u_{xt}}{2u_t^2} + \frac{u_xu_{yt}}{2u_t};\\
  X_{22}
  &= -\frac{f_2u_{yt}}{u_y}q + \frac{f_2u_t}{u_y}q_y,\\
  Y_{22}
  &= 0,\\
  T_{22}
  &= \frac{f_2(u_xu_{yt} + 3u_yu_{xt})}{u_yu_t}q - 3f_2q_x
    - \frac{f_2 u_x}{u_y}q_y;\\
  X_3
  &= \frac{3f_3(u_{yt}q - u_yq_t)}{u_t^5},\\
  Y_3 &= -\frac{f_3(3u_{xt}q + u_xq_t - u_tq_x)}{u_t^5},\\
  T_3 &= 0.
\end{align*}

\subsection{$A = -C$, $2B = 3C$}
\label{sec:a-=-c-9}

Cosymmetries:
\begin{gather*}
  \psi_1 = \frac{f_1u_t}{u_x^{7/2}};\\
  \psi_{21} = f_2u_t + \frac{f_{2, y}uu_t}{2u_y},\quad
  \psi_{22} = \frac{f_2u_t}{u_y};\\
  \psi_3 = f_3\sqrt{u_t}.
\end{gather*}

Conservation laws:
\begin{align*}
  X_1
  &= 0,\\
  Y_1
  &= \frac{f_1(3u_{xt}q - 2u_xq_t + 2u_tq_x)}{2u_x^{{7}/{2}}},\\
  T_1
  &= -\frac{3f_1(2u_xu_{yt}q - 3u_yu_{xt}q - 2u_tu_yq_x)}{4u_x^{{7}/{2}}u_t};\\
  X_{21}
  &= \frac{f_2(2uu_yu_{yyt} - (2uu_{yy} + 5u_y^2)u_{yt})}{4u_y^2}q
    + \frac{f_2u_t(2uu_{yy} + 5u_y^2)}{4u_y^2}q_y\\
  &- \frac{f_2uu_t}{2u_y}q_{yy},\\
  Y_{21}
  &= 0,\\
  T_{21}
  &= f_2 \alpha q - \frac{3f_2(uu_{yt} - u_tu_y)}{4u_t}q_x
    - \frac{f_2((4u_tu_x - 9uu_{xt})u_y^2 + 4uu_tu_xu_{yy})}{8u_tu_y^2}q_y\\
  &+ \frac{3f_2u(2u_xu_{yt} - 3u_yu_{xt})}{8u_t^2}q_t
    + \frac{f_2uu_x}{2u_y}q_{yy} - \frac{3f_2uu_x}{4u_t}q_{yt}
    + \frac{9f_2uu_y}{8u_t}q_{xt}
\end{align*}

\subsection{$A = -C$, $B = 3C$}
\label{sec:a-=-c-10}

Cosymmetries:
\begin{gather*}
  \psi_1 = \frac{f_1u_t}{u_x^5};\\
  \psi_{21} = f_2u_t +\frac{f_{2, y}uu_t}{2u_y},\quad
  \psi_{22} = f_2\frac{u_t}{u_y};\\
  \psi_3 = f_3u_t^2.
\end{gather*}

Conservation laws:
\begin{align*}
  X_1
  &= 0,\\
  Y_1
  &= \frac{f_1(3u_{xt} q - u_xq_t + u_tq_x)}{u_x^5},\\
  T_1
  &= -\frac{3f_1(u_xu_{yt} - 3u_yu_{xt})}{u_x^5u_t}q +
    \frac{3f_1u_y}{u_x^5}q_x;\\ 
  X_{21}
  &= \frac{f_2(uu_yu_{yyt}u - (uu_{yy} + 4u_y^2)u_{yt})}{2u_y^2}q
    + \frac{f_2u_t(uu_{yy} + 4u_y^2)}{2u_y^2}q_y
    - \frac{f_2uu_t}{2u_y}q_{yy},\\
  Y_{21} &= 0,\\
  T_{21}
  &= f_2\alpha q - \frac{3f_2(uu_{yt} - u_tu_y)}{2u_t}q_x
    - \frac{f_2 ((u_tu_x - 9uu_{xt})u_y^2 + uu_tu_xu_{yy})}{2u_y^2u_t}q_y\\
  &+ \frac{3f_2u(u_xu_{yt} - 3u_yu_{xt})}{2u_t^2}q_t
    + \frac{9f_2uu_y}{2u_t}q_{xt} + \frac{f_2uu_x}{2u_y}q_{yy}
    - \frac{3uf_2u_x}{2u_t}q_{yt},
    \intertext{where}
    \alpha
  &= -\frac{9uu_yu_{xtt}}{2u_t^2} - \frac{uu_xu_{yyt}}{2u_yu_t}
    + \frac{3uu_xu_{ytt}}{2u_t^2}
    - \frac{3u(u_xu_{yt} - 3u_yu_{xt})u_{tt}}{2u_t^3}\\
  &- \frac{3(2uu_{yt} + u_tu_y)u_{xt}}{2u_t^2}
    + \frac{u_x(uu_{yy} + u_y^2)u_{yt}}{2u_y^2u_t};\\
  X_{22}
  &= -\frac{f_2u_{yt}}{u_y}q + \frac{f_2 u_t}{u_y}q_y,\\
  Y_{22}
  &= 0,\\
  T_{22}
  &= \frac{f_2(u_xu_{yt} - 3u_yu_{xt})}{u_yu_t}q
    + 3f_2q_x - \frac{f_2 u_x}{u_y}q_y;\\
  X_3
  &= -3f_3u_tu_{yt}q + 3f_3u_tu_yq_t,\\
  Y_3
  &= 3f_3u_tu_{xt}q - f_3u_tu_xq_t + f_3u_t^2q_x\\
  T_3
  &= 0.
\end{align*}

\subsection{$3A = 2C$, $3B = -2C$}
\label{sec:3a-=-2c-1}

Cosymmetries:
\begin{gather*}
 \psi_1 = \frac{f_1u_t}{\sqrt{u_x}};\\
 \psi_2 = \frac{f_2 u_t}{u_y^{{7}/{2}}};\\
 \psi_{31} = f_3,\quad
 \psi_{32} = \frac{f_{3, t} u}{2} + f_3u_t.
\end{gather*}
Conservation laws:
\begin{align*}
  X_1
  &= 0,\\
  Y_1
  &= -\frac{2f_1u_{xt}}{3\sqrt{u_x}}q + \frac{f_1u_t}{\sqrt{u_x}}q_x
    + \frac{2f_1\sqrt{u_x}}{3}q_t,\\
  T_1
  &= -\frac{4f_1(u_xu_{yt} - u_yu_{xt})}{9\sqrt{u_x}u_t}q
    - \frac{2f_1u_y}{3\sqrt{u_x}}q_x;\\
  X_2
  &= -\frac{f_2u_{yt}}{u_y^{{7}/{2}}}q + \frac{f_2u_t}{u_y^{{7}/{2}}}q_y,\\
  Y_2
  &= 0,\\
  T_2
  &= -\frac{2f_2(u_xu_{yt} - u_yu_{xt})}{3u_y^{{7}/{2}}u_t}q
    - \frac{2f_2}{3u_y^{{5}/{2}}}q_x
    + \frac{2f_2u_x}{3u_y^{{7}/{2}}}q_y;\\
  X_{31}
  &= \frac{2f_3u_{yt}}{3u_t}q - \frac{2f_3u_y}{3u_t}q_t,\\
  Y_{31}
  &= -\frac{2f_3u_{xt}}{3u_t}q + f_3q_x + \frac{2f_3 u_x}{3u_t}q_t,\\
  T_{31}
  &= 0;\\
  X_{32}
  &= -\frac{f_3(4u_t^2 - 2uu_{tt})u_{yt} + 3u_t^2f_{3,t}u_y + 2 u_ytt u_t f_3
    u}{6u_t^2}q\\
  &- \frac{f_{3,t} u + f_3u_t}{2}q_y
    - \frac{f_3u_y(uu_{tt} + u_t^2)}{3u_t^2}q_t
    - \frac{f_3 u}{2}q_{yt} + \frac{f_3uu_y}{3u_t}q_{tt},\\
  Y_{32}
  &= \frac{f_3(u_t^2 - 2uu_{tt})u_{xt} + 3f_{3,t}u_t^2u_x  +
    2f_3uu_tu_{xtt}}{6u_t^2}q 
    + \frac{f_3u_x(2uu_{tt} + 5u_t^2)}{6u_t^2}q_t\\
  &+ \frac{f_{3, t}u + 2f_3u_t}{2}q_x - \frac{f_3 uu_x}{3u_t}q_{tt},\\
  T_{32}
  &= 0.
\end{align*}

\subsection{$3A = C$, $3B = -C$}
\label{sec:3a-=-c-1}

Cosymmetries:
\begin{gather*}
  \psi_1 = f_1u_xu_t,\\
  \psi_2 = \frac{f_2u_t}{u_y^5};\\
  \psi_{31} = f_3,\quad
  \psi_{32} = \frac{f_{3, t}u}{2} + f_3 u_t.
\end{gather*}
Conservation laws:
\begin{align*}
  X_1
  &= 0,\\
  Y_1
  &= -\frac{f_1u_xu_{xt}}{3}q + f_1u_xu_tq_x + \frac{f_1u_x^2}{3}q_t,\\
  T_1
  &= -\frac{f_1u_x(u_xu_{yt} - u_yu_{xt})}{9u_t}q - \frac{f_1u_xu_y}{3}q_x;\\
  X_2
  &= -\frac{f_2u_{yt}}{u_y^5}q + \frac{f_2u_t}{u_y^5}q_y,\\
  Y_2
  &= 0,\\
  T_2
  &=  -\frac{f_2(u_xu_{yt} - u_yu_{xt})}{3u_y^5u_t}q
    - \frac{f_2}{3u_y^4}q_x + \frac{f_2u_x}{3u_y^5}q_y;\\
  X_{31}
  &= \frac{f_3u_{yt}}{3u_t}q - \frac{f_3u_y}{3u_t}q_t,\\
  Y_{31}
  &= -\frac{f_3u_{xt}}{3u_t}q + f_3q_x + \frac{f_3u_x}{3u_t}q_t,\\
  T_{31}
  &=0;\\
  X_{32}
  &= \frac{f_3(uu_{tt} - 5u_t^2)u_{yt} - u_t(3f_{3,t}u_tu_y + f_3uu_{ytt})}{6
    u_t^2}q 
    - \frac{f_3u_t + f_{3,t} u}{2}q_y\\
  &- \frac{f_3u_y(uu_{tt} + u_t^2)}{6u_t^2}q_t - \frac{f_3u}{2}q_{yt}
    + \frac{f_3uu_y}{6u_t}q_{tt},\\
    Y_{32}
  &= \frac{f_3(2u_t^2 - uu_{tt})u_{xt} + u_t(3f_{3,t}u_tu_x +
    f_3uu_{xtt})}{6u_t^2}q 
    + \frac{f_3u_x(uu_{tt} + 4u_t^2)}{6u_t^2}q_t\\
  &+ \frac{f_{3,t} u + 2f_3u_t}{2}q_x - \frac{f_3uu_x}{6u_t}q_{tt},\\
    T_{32}
  &= 0.
\end{align*}

\subsection{$A = C$, $2B = -C$}
\label{sec:a-=-c-11}

Cosymmetries:
\begin{gather*}
  \psi_0 = u_t,\\
  \psi_1 = \frac{f_1u_t}{\sqrt{u_x}},\\
  \psi_2 = \frac{f_2u_t}{u_y^5},\\
  \psi_3 = \frac{f_3}{\sqrt{u_t}}.
\end{gather*}
Conservation laws:
\begin{align*}
  X_0
  &= -\frac{u_y}{2}q_t,\\
  Y_0
  &= -\frac{u_{xt}}{2}q + u_xq_t + u_tq_x,\\
  T_0
  &= -\frac{2u_xu_{yt} - u_yu_{xt}}{4u_t}q;\\
  X_1
  &= 0,\\
  Y_1
  &= f_1\left(-\frac{u_xt}{2\sqrt{u_x}}q + \frac{u_t}{\sqrt{u_x}}q_x
    + \sqrt{u_x} q_t\right),\\
  T_1
  &= -\frac{f_1(2u_xu_{yt} - u_yu_{xt})}{4\sqrt{u_x}u_t}q
    - \frac{f_1u_y}{2\sqrt{u_x}}q_x;\\
  X_2
  &= \frac{f_2u_{yt}}{u_y^5}q + \frac{f_2u_t}{u_y^5}q_y
    - \frac{f_2}{2u_y^4}q_t,\\
  Y_2 &= 0,\\
  T_2
  &= \frac{f_2(2u_xu_{yt} - u_yu_{xt}u_y)}{2u_y^5u_t}q
    + \frac{f_2u_x}{u_y^5}q_y;\\
  X_3
  &= -f_3\left(\frac{u_y}{2\sqrt{u_t}}q_t - \frac{u_{yt}}{\sqrt{u_t}}q
    - \sqrt{u_t}q_y\right),\\
  Y_3
  &= -\frac{f_3u_{xt}}{\sqrt{u_t}}q + \frac{f_3u_x}{\sqrt{u_t}}q_t,\\
  T_3
  &= 0.
\end{align*}

\subsection{$A = -2C$, $B = -2C$}
\label{sec:a-=-2c-1}

Cosymmetries:
\begin{gather*}
  \psi_0 = u_t;\\
  \psi_1 = \frac{f_1u_t}{\sqrt{u_x}};\\
  \psi_2 = \frac{f_2u_t}{\sqrt{u_y}};\\
  \psi_3 = \frac{f_3}{u_t^4}
\end{gather*}
Conservation laws:
\begin{align*}
  X_0
  &= -3u_{yt}q - 2u_yq_t,\\
  Y_0 &= -2u_{xt}q - 2u_xq_t + u_tq_x,\\
  T_0
  &= \frac{10(u_xu_{yt} + u_yu_{xt})}{u_t}q;\\
  X_1
  &= 0,\\
  Y_1
  &= -\frac{2f_1u_{xt}}{\sqrt{u_x}}q + \frac{f_1u_t}{\sqrt{u_x}}q_x
    - 2f_1q_t\sqrt{u_x},\\
  T_1 &= \frac{4f_1(u_xu_{yt} + u_yu_{xt})}{\sqrt{u_x}u_t}q
        -\frac{2f_1u_y}{\sqrt{u_x}}q_x;\\
  X_2
  &= -\frac{f_2u_{yt}}{\sqrt{u_y}}q + \frac{f_2u_t}{\sqrt{u_y}}q_y,\\
  Y_2
  &= 0,\\
  T_2
  &= \frac{2f_2(u_xu_{yt} + u_yu_{xt})}{\sqrt{u_y}u_t}q
    - 2f_2\sqrt{u_y} q_x - \frac{2f_2u_x}{\sqrt{u_y}}q_y;\\
  X_3
  &= \frac{2f_3u_{yt}}{u_t^5}q - \frac{2f_3u_y}{u_t^5}q_t,\\
  Y_3 &= -\frac{f_3(2u_{xt}q + 2u_xq_t - u_tq_x)}{u_t^5},\\
  T_3
  &= 0.
\end{align*}

\end{document}